\documentclass[submission,copyright,creativecommons]{eptcs}
\providecommand{\event}{NCMA 2024} % Name of the event you are submitting to

\usepackage{iftex}

\ifpdf
  \usepackage{underscore}         % Only needed if you use pdflatex.
  \usepackage[T1]{fontenc}        % Recommended with pdflatex
\else
  \usepackage{breakurl}           % Not needed if you use pdflatex only.
\fi

\usepackage{amsthm}
\usepackage{amsfonts}
\usepackage{amsmath}
\usepackage{amssymb}
\usepackage{color}
\usepackage{easybmat}
\usepackage{blkarray}
\usepackage{datetime}
\usepackage{hyperref}
\usepackage{comment}
\usepackage{bm}

\usepackage{mathrsfs}
\renewcommand{\mathcal}[1]{\mathscr{#1}}

\usepackage{marvosym}

\theoremstyle{plain}
\newtheorem{theorem}{Theorem}%[section]
\newtheorem{lemma}[theorem]{Lemma}
\newtheorem{claim}{Claim}%[theorem]
\newcommand{\FF}{\mathcal{F}}
\newcommand{\CC}{\mathcal{C}}
\newcommand{\inneigh}{\mathrm{N}_{-}}
\newcommand{\outneigh}{\mathrm{N}_{+}}

\renewcommand{\phi}{\varphi}

\newcommand\restr[2]{{% we make the whole thing an ordinary symbol
  \left.\kern-\nulldelimiterspace % automatically resize the bar with \right
  #1 % the function
  \vphantom{\big|} % pretend it's a little taller at normal size
  \right|_{#2} % this is the delimiter
}}

\renewcommand{\restr}[2]{#1{\uparrow}_{#2}}

\newcommand{\bfSig}{\boldsymbol{\Sigma}}

\ddmmyyyydate
\title{Complexity Aspects of the Extension of\\ Wagner's Hierarchy to $k$-Partitions}

\author{Vladimir Podolskii
\institute{Tufts University, Medford, MA, USA} \email{podolskii.vv@gmail.com}
\and Victor  Selivanov
\institute{Department of Mathematics and Computer Science, St.Petersburg  University, Saint Petersburg, Russia}
\institute{A.P. Erhov Institute of Informatics Systems, Novosibirsk, Russia}
\email{vseliv@iis.nsk.su}
}
\def\authorrunning{V. Podolskii, V. Selivanov}
\def\titlerunning{Complexity Aspects of the Extension of Wagner's Hierarchy to $k$-Partitions}

\date{}
\begin{document}

\maketitle

\begin{abstract}
    It is known that the Wadge reducibility of regular $\omega$-languages is efficiently decidable~(Krishnan et al., 1995), (Wilke, Yoo, 1995).

    In this paper we study analogous problem for regular $k$-partitions of $\omega$-languages. In the series of previous papers (Selivanov, 2011), (Alaev, Selivanov, 2021), (Selivanov, 2012) there was a partial progress towards obtaining an efficient algorithm for deciding the Wadge reducibility in this setting as well. In this paper we finalize this line of research providing a quadratic algorithm (in RAM model). For this we construct a quadratic algorithm to decide a preorder relation on iterated posets. 

    Additionally, we discuss the size of the representation of regular $\omega$-languages and suggest a more compact way to represent them. The algorithm we provide is efficient for the more compact representation as well.

    %{\em Key Words.} Wadge reducibility, Wagner hierarchy, $k$-partition, Muller $k$-acceptor, iterated labeled poset, efficient algorithm.

    %\keywords{Wadge reducibility  \and Wagner hierarchy \and $k$-partition \and Muller $k$-acceptor \and iterated labeled poset \and efficient algorithm.}
\end{abstract}

\section{Introduction}\label{in}

In \cite{wag79}, K. Wagner has shown that the quotient-poset of the preorder ${(\mathcal{R};\leq_W)}$ of regular $\omega$-languages under the Wadge reducibility (i.e., $m$-reducibility by continuous functions on the Cantor space of $\omega$-words) is semi-well-ordered with order type $\omega^\omega$, and that the related algorithmic problems are decidable. E.g., given Muller acceptors $\mathcal{A}$ and $\mathcal{B}$, one can effectively solve the relation $L(\mathcal{A})\leq_W L(\mathcal{B})$ between the corresponding regular $\omega$-languages. 
Later it was shown that there are efficient algorithms solving such problems   \cite{kpb,wy}, in particular the problem $L(\mathcal{A})\leq_W L(\mathcal{B})$ is solvable in cubic time.%\nb{Is it on Turing Machines? No, they use essentially the same model. Cube might appear because they consider precisely Muller acceptance which contains less information than our modification for 2-partitions (they have only one row of the acceptance table of our two). For k>2 deletion on one row is also possible but less natural as for k=2. As a result, our algorithms does not generalize theirs, and we do not claim this. Originally, I wondered on the existence of a polynomial-time algorithm:)}

In \cite{s11} (see also \cite{s23} for detailed proofs), the Wagner theory was extended from the regular  sets $A\subseteq X^\omega$ (identified in the usual way with functions ${A:X^\omega\rightarrow \{0,1\}}$) to the regular $k$-partitions $A:X^\omega\rightarrow \{0,\ldots,k-1\}$ of the set $X^\omega$ of $\omega$-words over a finite alphabet $X$.  Motivations for this extension
come from the fact that similar objects are important e.g. in computability theory~\cite{s22}, descriptive set theory~\cite{km19},  and complexity theory~\cite{ko00}. 

The extension from sets to $k$-partitions for $k>2$ is  non-trivial in the sense that the corresponding structure $(\mathcal{R}_k;\leq_W)$ becomes much more complex. Nevertheless, it admits a nice combinatorial characterization in terms of iterated $h$-preorders on labeled forests (terminology is briefly  recalled in the next section), and the full extension of  Wagner's hierarchy to $k$-partitions is possible; the extension is called the fine hierarchy  (FH) of $\omega$-regular $k$-partitions). But the existence of the corresponding {\em efficient} algorithms for the algorithmic problems (like the extension of $L(\mathcal{A})\leq_W L(\mathcal{B})$ to Muller's $k$-acceptors \cite{s23}) is far from obvious.

In this paper, we address the latter problem. A first step in this direction was made in \cite{as21} where, with the use of some previous results from \cite{hs14}, efficient algorithms deciding basic problems about the iterated $h$-preorders on labeled forests were established. This is relevant because levels of the FH of $k$-partitions are naturally denoted by the iterated labeled forests, and manipulations with the levels seem inevitable. But unfortunately, this does not immediately yield an efficient algorithm for solving $L(\mathcal{A})\leq_W L(\mathcal{B})$ because one needs first to find (from given Muller's $k$-acceptors) the levels of the FH (i.e., the corresponding forests $F$ for $\mathcal{A}$ and $G$ for $\mathcal{B}$) where the $k$-partitions $L(\mathcal{A})$ and $L(\mathcal{B})$ are Wadge complete. In computing $F,G$ from $\mathcal{A},\mathcal{B}$, one can first compute iterated $k$-posets $P,R$ (this computation is feasible), and then to unfold $P,R$ to forests $F=u(P),G=u(R)$ using a natural algorithm first described  in \cite{s04} and then elaborated in \cite{s12}. Unfortunately, the size of the unfolded forests grows exponentially, so on this way it is hopeless to find an efficient algorithm.

A possible solution is suggested by the observation in \cite{s12} that it is possible to name levels of the FH of $k$-partitions directly by the iterated poset $P$ instead of the unfolded forest $u(P)$, obtaining thus a more succinct notation system for levels. But then we have to work with the preorder $\preceq$ on posets induced by the preorder $\leq_h$ (i.e., $P\preceq R$ iff $u(P)\leq_hu(R)$), and we cannot directly use the algorithms from \cite{as21}. 

As the main result of this paper we reprove the complexity estimates in~\cite{as21} directly for posets. This result applies not only to the extension of Wagner's hierarchy but might also be useful in other situations where the FH of $k$-partitions naturally appears (for an example from computability theory see \cite{s22}). With this at hand, it is not hard to get the desired efficient algorithms for the extended Wagner hierarchy.

For our algorithms we use the RAM model, which is more standard for studying efficient algorithms than the Turing machine model used in \cite{as21}. In this model our algorithm is quadratic and we note that the algorithm from~\cite{as21} is quadratic as well in the RAM model (see Section~\ref{sec:model} for details). 

Since we are studying efficient algorithms, the size of the representation of the input matters. The straightforward way to represent $k$-acceptors is to provide for each subset of states a label in the partition. However, note that Muller acceptors operate with cycles and not all subsets can be cycles. We observe that actually the number of cycles is polynomially smaller than the number of subsets of states. This suggests another way to represent $k$-acceptors that is more compact and might be useful in some settings. We note that the algorithm we provide is efficient  for this type of representation of inputs as well.

%Fortunately, it is possible to reprove the complexity estimates in \cite{as21} directly for posets (though in a slightly different computational model, rather that Turing machines used in \cite{as21}). This is probably the main result of this paper, which applies not only to the extension of Wagner's hierarchy but might also be useful in other situations where the FH of $k$-partitions naturally appears (for an example from computability theory see \cite{s22}). With this at hand, it is not hard to get the desired efficient algorithms for the extended Wagner hierarchy.

After recalling some preliminaries in the next section, in Section \ref{poset} we prove the mentioned result on iterated labeled posets. In Section \ref{accept} we explain how to deduce efficient algorithms working with a straightforward representation of Muller's $k$-acceptors. 
In Section~\ref{sec:representation} we discuss a less straightforward representations of acceptors and translate our algorithm for Muller's $k$-acceptors to a more succinct representation of inputs.

%We conclude in   Section \ref{con} with a short discussion of other related algorithms. 

\section{Preliminaries}\label{pre}

We use standard notation and facts about finite automata on infinite words which may be found e.g. in \cite{pp,th90}. We work with a fixed finite  alphabet $X$ containing more than one letter, and only with deterministic finite automata.

\subsection{Automata and acceptors}
By an {\em automaton} (over $X$) we mean a triple ${\mathcal M}=(Q,f,in)$
consisting  of a finite non-empty set $Q$ of states,  a transition
function $f:Q\times X\rightarrow Q$ and an initial state $in\in Q$.
The  function $f$ is  extended to the function
$f:Q\times X^*\rightarrow Q$  by induction
$f(q,\varepsilon)=q$ and $f(q,u\cdot x)=f(f(q,u),x)$, where $u\in
X^*$ and $x\in X$. Similarly, we may define the function $f:Q\times
X^\omega\rightarrow Q^\omega$ by $f(q,\xi)(n)=f(q,\xi\upharpoonright_n)$.

Associate
with any automaton ${\mathcal M}$ the set of {\em cycles} (known also as loops) $C_{\mathcal M}=\{f_{\mathcal
M}(\xi)\mid\xi\in X^\omega\}$ where $f_{\mathcal M}(\xi)$ is the set of
states that occur infinitely often in the sequence $f(in,\xi)\in
Q^\omega$.
A {\em Muller acceptor} is a pair $({\mathcal M},{\mathcal F})$ where ${\mathcal M}$ is an
automaton and ${\mathcal F}\subseteq C_{\mathcal M}$; it recognizes the set
$ L({\mathcal M},{\mathcal F})=\{\xi\in X^\omega\mid f_{\mathcal M}(\xi)\in{\mathcal
F}\}$. The Muller  acceptors recognize exactly the
{\em regular $\omega$-languages}. 

A $k$-partition $A:X^\omega\rightarrow \{0,\ldots,k-1\}$ is {\em regular}, if any its component $A_i=A^{-1}(i)$, $i<k$, is regular. Regular $k$-partition $A$ may be represented by $k$-tuples of Muller acceptors which recognize the components of $A$, but we will use a slightly different kind of acceptors introduced in \cite{s11}. A {\em Muller  $k$-acceptor} is a pair $({\mathcal M},A)$ where
${\mathcal M}$ is an  automaton and $A: C_{\mathcal M}\rightarrow k$ is a
$k$-partition of $C_{\mathcal M}$. The Muller  $k$-acceptor $({\mathcal M},A)$ recognizes the 
$k$-partition $L({\mathcal M},A)=A\circ f_{\mathcal M}$ where $f_{\mathcal
M}:X^\omega\rightarrow C_{\mathcal M}$ is  defined above.

Note that the Muller  $2$-acceptors are equivalent to Muller acceptors though syntactically they are slightly different beecause, along with the set $\mathcal{F}$ of accepting cicles a Muller  $2$-acceptor also contains its complement $C_\mathbf{M}\setminus\mathcal{F}$. This causes some distinctions of our complexity estimates for $k=2$ from those in \cite{kpb,wy}.

\subsection{Iterated $k$-posets}

Next we recall some information about the iterated $h$-preorder and its variants; for additional information see e.g. \cite{s22,s23}. Let $(P;\leq)$ be a finite poset; if $\leq$ is clear from the context, we simplify the notation of the poset to $P$. Any subset of $P$ may be considered as a poset with the induced
partial ordering.  By a {\em  forest} we mean a finite poset  in which every lower cone $\downarrow{x}$, $x\in P$, is a chain. A {\em  tree} is a forest with the least element (called the {\em  root} of the
tree). 

%The ``abstract'' trees (forests) just defined are for almost all purposes equivalent to their isomorphic copies realised as initial segments of $(\omega^*;\sqsubseteq)$ (resp. $(\omega^+;\sqsubseteq)$) where $\sqsubseteq$ is the prefix relation on finite strings of naturals. Below we often work with such ``concrete'' copies which enable to use convenient standard notation for strings.

Let $(Q;\leq)$ be a preorder. A {\em $Q$-poset}  is a triple
$(P,\leq,c)$ consisting of a finite nonempty poset $(P;\leq)$,
$P\subseteq\omega$,  and a labeling $c:P\rightarrow Q$. Let ${\mathcal P}_Q$, ${\mathcal
F}_Q$, and ${\mathcal T}_Q$  denote the sets of
all finite $Q$-posets, $Q$-forests, and $Q$-trees, respectively.
For the particular case $Q=\bar{k}=\{0,\cdots,k-1\}$ of
antichain with $k$ elements we denote the corresponding $h$-preorders by ${\mathcal
P}_k$, ${\mathcal F}_k$, and ${\mathcal T}_k$. A {\em morphism}
${f:(P,\leq,c)}\rightarrow(P^\prime,\leq^\prime,c^\prime)$ between
$Q$-posets is a monotone function
$f:(P;\leq)\rightarrow(P^\prime;\leq^\prime)$ satisfying $\forall
x\in P\,(c(x)\leq c^\prime(f(x)))$. 
The {\em $h$-preorder} $\leq_h$ on ${\mathcal P}_Q$ is defined as
follows: $P\leq_h P^\prime$, if there is a  morphism
$f:P\rightarrow P^\prime$.  

For any $q\in Q$ let $s(q)\in{\mathcal T}_Q$ be the singleton tree labeled by $q$; then $q\leq r$ iff $s(q)\leq_hs(r)$. Identifying $q$ with $s(q)$, we may think that $Q$ is a substructure of ${\mathcal T}_Q$.
The quotient-poset of $({\mathcal F}_Q;\leq_h,\sqcup)$ is a semilattice where the supremum operation is induced by the disjoint union $F\sqcup G$ of $Q$-forests $F,G$. The semilattice is generated by the join-irreducible elements induced by trees. The set ${\mathcal F}_Q$ may be identified with the set ${\mathcal T}^\sqcup_Q$ of finite disjoint unions of trees. 

For any finite $Q$-poset $(P,c)$ there
exist a finite $Q$-forest $(F,d)$ and a morphism $f$ from $F$
onto $P$ such that $F$ is
a largest element in ${(\{G\in{\mathcal F}_Q\mid
G\leq_hP\};\leq_h)}$. The forest $F=u(P)$ is constructed by a
natural bottom-up unfolding of $P$ (for additional details see
\cite{s04} and sections 7,8 of \cite{s23}). The unfolding operator $u:\mathcal{P}_Q\to\mathcal{F}_Q$ gives rise to a preorder $\preceq$ on $\mathcal{P}_Q$ (already mentioned in the introduction) defined by $P\preceq R\leftrightarrow u(P)\leq_hu(R)$. Note that $P\leq_h R$ implies $P\preceq R$ (but the converse fails in general), and that both relations coincide on $\mathcal{F}_Q$.

Define the sequence $\{\mathcal{T}_k(n)\}_{n<\omega}$ of preorders by induction on $n$
as follows: $\mathcal{T}_k(0)=\overline{k}$  and
$\mathcal{T}_k(n+1)=\mathcal{T}_{\mathcal{T}_k(n)}$. The sets $\mathcal{T}_k(n)$, $n<\omega$, are pairwise disjoint but, identifying
the elements $i$ of $\overline{k}$ with the corresponding singleton trees $s(i)$   labeled by $i$ (which are precisely the
minimal elements  of $\mathcal{T}_k(1)$), we may think that
$\mathcal{T}_k(0)\sqsubseteq\mathcal{T}_k(1)$, i.e. the quotient-poset of the first preorder is an initial segment of the quotient-poset of the second. This also induces an embedding of $\mathcal{T}_k(n)$ into $\mathcal{T}_k(n+1)$ as an initial segment, so (abusing notation) we may think that $\mathcal{T}_k(0)\sqsubseteq\mathcal{T}_k(1)\sqsubseteq\cdots$.

Let
$\mathcal{T}_k(\omega)=\bigcup_{n<\omega}\mathcal{T}_k(n)$; the induced preorder on this set is again denoted by $\leq_h$. We often simplify $\mathcal{T}_k(n)^\sqcup$ to $\mathcal{F}_k(n)$; in particular, $\mathcal{F}_k(2)=\mathcal{F}_{\mathcal{T}_k}$. The embedding $s$ is extended to $\mathcal{T}_k(\omega)$ by defining $s(T)$ as the singleton tree labeled by $T$.

Similar iterations are possible for  other aforementioned constructions. E.g., we can define iterations of the construction $Q\mapsto\mathcal{V}_Q$ where $\mathcal{V}_Q$ is the set of pointed posets from $\mathcal{P}_Q$ (i.e., posets with a smallest element) ordered by the relation $\preceq$ (rather than by $\leq_h$). In this way, we obtain the sequence $\{\mathcal{V}_k(n)\}_{n\leq\omega}$.   

The aforementioned unfolding operator $u:\mathcal{P}_Q\to\mathcal{F}_Q$ is naturally extended and modified to operators $u:\mathcal{V}_k(n)\to\mathcal{T}_k(n)$ and $u:\mathcal{V}_k(n)^\sqcup\to\mathcal{T}_k(n)^\sqcup$ for each $n\leq\omega$ which have properties similar to those of the basic operator $u:\mathcal{P}_Q\to\mathcal{F}_Q$ (cf. Lemma 8.7 in \cite{s23}). Especially relevant to this paper is the unfolding $u:\mathcal{P}_{\mathcal{V}_k}\to\mathcal{F}_{\mathcal{T}_k}$ which first unfolds the poset to a forest, and then unfolds the labels (which are pointed $k$-posets) to $k$-trees.

\subsection{Bases and fine hierarchies}

Next we recall some notation and notions relevant to the fine hierarchies. A {\em 1-base} in a set $S$ is just a subalgebra   $\mathcal{L}$ of $(P(S);\cup,\cap,\emptyset,S)$, i.e. a subset of of the Boolean $P(S)$ closed under finite unions and intersections. A {\em 2-base} in  $S$ is a pair $\mathcal{L}=(\mathcal{L}_0,\mathcal{L}_1)$ of 1-bases in $S$ such that $\mathcal{L}_0\cup\check{\mathcal{L}}_0\subseteq\mathcal{L}_{1}$, where $\check{\mathcal{L}}_1$ is the set of complements of sets in $\mathcal{L}_1$. 

By an {\em $\omega$-base in a set $S$} we mean a sequence $\mathcal{L}=\mathcal{L}(S)=\{\mathcal{L}_n\}_{n<\omega}$ of 1-bases such that  $\mathcal{L}_n\cup\check{\mathcal{L}}_n\subseteq\mathcal{L}_{n+1}$ for each $n$. Note that the $\omega$-bases subsume the 1-bases $\mathcal{L}$ (by taking $\mathcal{L}_0=\mathcal{L}$ and $\mathcal{L}_{k+1}=(\mathcal{L})$ where $(\mathcal{L})$ is the Boolean closure of $\mathcal{L}$) and the 2-bases $(\mathcal{L}_0,\mathcal{L}_1)$ (by taking $\mathcal{L}_0=\mathcal{L}_0$, $\mathcal{L}_1=\mathcal{L}_1$ and $\mathcal{L}_{k+2}=(\mathcal{L}_1)$). 

The $\omega$-base $\mathcal{L}$ is {\em reducible} if every its level $\mathcal{L}_n$ has the reduction property, i.e. for any $A,B\in\mathcal{L}_n$ there exist $A',B'\in\mathcal{L}_n$ such that $A'\subseteq A$, $B'\subseteq B$,  $A'\cap B'=\emptyset$, and  $A'\cup B'= A\cup B$.  

We give two examples of bases. Let $\{\bfSig^0_{1+n}\}_{n<\omega}$  be the $\omega$-base of finite $\bfSig$-levels of Borel hierarchy in the Cantor space $X^\omega$. This base is well known to be reducible.
The class $\mathcal{R}$  of regular $\omega$-languages over $X$ induces the $\omega$-base 
$\{\mathcal{R}\cap\bfSig^0_{1+n}\}_{n<\omega}$ in $\mathcal{R}$. 
Since all regular $\omega$-languages  sit in the Boolean closure of $\bfSig^0_2$, the latter $\omega$-base coincides with  the 2-base $\mathcal{R}\bfSig=(\mathcal{R}\cap\bfSig^0_1,\mathcal{R}\cap\bfSig^0_2)$. As shown in \cite{s98}, this base in also reducible.

With any $\omega$-base $\mathcal{L}$ in $S$ one can associate the {\em FH of $k$-partitions over $\mathcal{L}$} which is a family $\{\mathcal{L}(F)\}_{F\in\mathcal{T}_k(\omega)^\sqcup}$ of subsets of $k^S$. 
The notation system  $\mathcal{T}_k(\omega)^\sqcup$ for levels of the FH based on iterated trees and forests is convenient for establishing properties of the FH in a series of papers of the second author (see e.g. \cite{s23} and references therein). 

We do not reproduce here all (rather technical) details concerning the FH but we note that, instead of the notation system $\mathcal{T}_k(\omega)^\sqcup$ for its levels and the relation $\leq_h$ on the forests, we can equivalently take the larger system $\mathcal{V}_k(\omega)^\sqcup$ and the relation $\preceq$ on labeled posets, as explained in the introduction. In the second approach (first described in sections 7 and 8 of \cite{s12}) the FH over $\omega$-base $\mathcal{L}$ takes the form $\{\mathcal{L}(P)\}_{P\in\mathcal{V}_k(\omega)^\sqcup}$,
where the levels have the property: $\mathcal{L}(P)=\mathcal{L}(u(P))$ for each $P\in\mathcal{V}_k(\omega)^\sqcup$ (see Lemma 8.16(5) in \cite{s12}). This property is crucial for the results in this paper, as explained in the introduction. 

We give some details for the particular cases of FHs of $k$-partitions over 1-bases and 2-bases (which are in fact sufficient for this paper). The FH over a 1-base $\mathcal{L}$ in $S$ looks as $\{\mathcal{L}(F)\}_{F\in\mathcal{F}_k}$ (in the forest notation system)  or as $\{\mathcal{L}(P)\}_{P\in\mathcal{P}_k}$ (in the poset notation system). The level $\mathcal{L}(P)$ of the latter hierarchy consists of all $k$-partitions $A:S\to\bar{k}$ such that for some family $\{B_p\}_{p\in P}$ of $\mathcal{L}$-sets we have: $A_i=\bigcup\{\tilde{B}_p\mid p\in P_i\}$ for each $i<k$, where $\tilde{B}_p=B_p\setminus\bigcup\{B_q\mid p<q\}$ and $P_i=c^{-1}(i)$, $c:P\to\bar{k}$. According to section 7 of \cite{s12}, $\mathcal{L}(P)=\mathcal{L}(u(P))$, the family $\{B_p\}_{p\in P}$ may be assumed monotone (i.e., $B_p\supseteq B_q$ for $p<q$), and, if $P\in\mathcal{F}_k$ and $\mathcal{L}$ is reducible, the family $\{B_p\}_{p\in P}$ may be assumed reduced (i.e., $B_p\cap B_q=\emptyset$ for all incomparable $p,q\in P$). 

The FH over a 2-base $(\mathcal{L}_0,\mathcal{L}_1)$ in $S$ looks as $\{\mathcal{L}(F)\}_{F\in\mathcal{F}_{\mathcal{T}_k}}$ (in the forest notation system)  or as $\{\mathcal{L}(P)\}_{P\in\mathcal{P}_{\mathcal{V}_k}}$ (in the poset notation system). The level $\mathcal{L}(P)$ of the latter hierarchy consists of all $k$-partitions $A:S\to\bar{k}$ such that for some family $\{B_p\}_{p\in P}$ of $\mathcal{L}_0$-sets and  for some families $\{B_{p_0p_1}\}_{p_1\in c(p_0)}$  of $\mathcal{L}_1$-sets, $p_0\in P$,  we have: $A_i=\bigcup\{\tilde{B}_{p_0p_1}\mid p_0\in P, p_1\in c(p_0)_i\}$ for each $i<k$, where $\tilde{B}_{p_0p_1}=B_{p_0p_1}\setminus\bigcup\{B_{p_0q_1}\mid p_1<q_1\in c(p_0)\}$, $\tilde{B}_{p_0}=\bigcup\{\tilde{B}_{p_0p_1}\mid p_1\in c(p_0)\}$, and $c(p_0)_i=d^{-1}(i)$, $d:c(p_0)\to\bar{k}$. According to section 8 of \cite{s12}, $\mathcal{L}(P)=\mathcal{L}(u(P))$, the families above may be assumed monotone, and, if $P\in\mathcal{F}_{\mathcal{T}_k}$ and the base $(\mathcal{L}_0,\mathcal{L}_1)$ is reducible, the families above may be assumed reduced.

For the 2-base  $\mathcal{L}=(\mathcal{R}\cap\bfSig^0_1,\mathcal{R}\cap\bfSig^0_2)$ above,  the FH  $\{\mathcal{L}(F)\}_{F\in\mathcal{T}_k(2)^\sqcup}$ is the extension of Wagner's hierarchy to $k$-partitions introduced in \cite{s11}. For $k=2$ we get back to the classical Wagner hierarchy in a set-theoretical presentation from \cite{s98}. As shown in \cite{s11,s23}, every $\omega$-regular $k$-partition is Wadge complete in some level $F\in\mathcal{T}_k(2)^\sqcup$, all the possibilities are realized and the structure of such degrees is isomorphic to $\mathcal{T}_k(2)^\sqcup$.

\subsection{Computational model} \label{sec:model}

We conclude this section with comparing the computational model used in this paper with that from \cite{as21}.
The paper~\cite{as21} used Turing machines to analyze the complexity of their algorithm. However, for efficient algorithms the model that is closer to practical computations and that is widely considered to be standard is the Random Access Machine (RAM) model. Each memory entry in this model contains a string (typically bounded in length by $O(\log n)$, where $n$ is the size of input) and standard arithmetic operations on memory entries can be performed in constant time. Operations with memory (store, copy, load) as well as control operations (branching, subroutine calls) can also be done in constant time. See e.g.~\cite[Section 2.2]{clrs22} or~\cite[Section 1.1.2]{gt14} for more details.

The paper~\cite{as21} had constructed a cubic algorithm for the case of trees in the model of multitape Turing machines. We note that this algorithm is quadratic in RAM model. Indeed, the main recurrence relations on the complexity $t(n,k)$ of the algorithm given two structures of size $n$ and $k$ respectively are
$$
t(n,k) \leq \sum_{i=1}^{l} t(n_i,k) + O(n+k),
$$
where $n_1, \ldots, n_l$ are any natural numbers such that $\sum_{i}n_i = n$, and the symmetric relation for $k$ instead of $n$. The solution to this recurrence is $t(n,k) = {O(n^2k + nk^2)}$.
	
The term $O(n+k)$ is needed on the step of the recursive construction to scan the inputs to find labels of specific nodes. This requires linear complexity on multitape Turing machines, but can be done with constant number operations in our model. Thus in our model the recurrence changes to
$$
t(n,k) \leq \sum_{i=1}^{l} t(n_i,k) + O(1)
$$
and the complexity drops to $t(n,k) = O(nk)$.

%We would like to note that in our algorithm as elementary operations with numbers we will use only comparisons. As we described above, they are assumed to take $O(1)$ time in RAM model. In case we would like to consider a model that allows operations only with bits, rather then integers, the complexity of comparison becomes linear in the size of the binary representation of numbers. However, we also have to consider the input integers to be given in binary representation and the running time bound is still the same w.r.t. the size of input.

\section{Algorithms on labeled posets}\label{poset}

Suppose we are given a finite poset $(P,c)$ labeled by a poset $Q$, i.e. $(P,c)\in\mathcal{P}_Q$ (see Section \ref{pre}). 
In this section it is convenient to represent the bottom-up unfolding of $P$ by $F(P)$ instead of $u(P)$.
It is convenient to represent elements of the unfolding $F(P)$ as paths $v = v_1\ldots v_t$, where $v_i \in P$, $v_1$ is a minimal element in $P$, $v_i < v_{i+1}$, and there are no other elements of $P$ between $v_i$ and $v_{i+1}$. The ordering on $F(P)$ can be naturally described as $v \leq w$, if $v$ is a prefix of $w$. 
%It is not hard to see that $F(P)$ is a forest. We can naturally transfer labeling from $P$ to $F(P)$. For the sake of convenience we will use the same notation for labeling function for $P$ and $F(P)$.

According to Section \ref{pre}, on $\mathcal{P}_Q$ we have two preorders: $\leq_h$ and $\preceq$, where $P_1\preceq P_2$  means that $u(P_1)\leq_h u(P_2)$. Deciding the relation $\leq_h$ on $\mathcal{P}_k$ is NP-complete for every $k\geq2$~\cite{kl}. It becomes polynomial time decidable if we restrict posets to forests. However, as we show in this section, the order $\preceq$ is polynomial time decidable for arbitrary posets. 

For $v \in P$ denote by $\restr{P}{v} \subseteq P$ an upper cone of $v$ in $P$. It is easy to see that posets $F(\restr{P}{v})$ and $\restr{F(P)}{u}$, where $u$ is an arbitrary element of unfolding with an endpoint in $v$, are isomorphic.

From this the following observation follows easily.

\begin{lemma}
For any element $v \in P$ all orders $\restr{F(P)}{u}$, such that $u$ has an endpoint $v$, are isomorphic.
\end{lemma}

Now we are ready to prove the main result of this section.

\begin{theorem}\label{pos}
Assume that there is an algorithm $A$ that checks the relation ${q_1\leq q_2}$ on $Q$ in time $C \cdot |q_1|\cdot |q_2|$, where $|q_i|$ is the size of the description of $q_i$ and $C$ is a positive constant. Then, there is an algorithm that checks the relation $(P_1,c_1)\preceq(P_2,c_2)$ on $\mathcal{P}_Q$ in time $C \cdot |P_1|\cdot |P_2|$.
\end{theorem}

\begin{proof}
	%We will consider a computational model that is standard for algorithms theory, in which certain basic operations operations are performed in $O(1)$ time. Among basic operations there are comparisons of numbers, following a pointer and others\nb{This is a RAM model, discuss this earlier}. In this model the standard loops and if-else operators are allowed.

	Next we describe the algorithm.
	
	We will compute for all pairs of elements $v_1 \in P_1$ and $v_2 \in P_2$ one bit of information:
	whether there is a morphism from $F(\restr{P_1}{v_1})$ to $F(\restr{P_2}{v_2})$. We denote this bit by $M(v_1, v_2)$.
	
	We assume that on the input we are given graphs corresponding to posets in which directed edges connect an element $p$ to each element $q$ such that $p < q$ and there are no other elements between $p$ and $q$. We call $q$ a \emph{successor} of $p$. The graphs are given as adjacency lists.
	
	We run a depth-first search (DFS) on $P_1$. For each vertex $v_1$ we fill-in the corresponding row of $M$ after visiting all its successors (this is our post-processing of the vertex $v_1$ in the DFS~\cite[Section 3.2.1]{dpv08}). For this we run a depth-first search on $P_2$ (that is, we have a loop over all vertices of $P_1$ and inside of it another loop over all vertices of $P_2$). 
	
	After visiting all the neighbors of a vertex $v_2$ we can use the following to compute $M(v_1,v_2)$. 
	
	\begin{claim} %\label{cl:dp-step}
		$M(v_1,v_2)=1$ iff there is a successor $u_2$ of $v_2$, such that $M(v_1,u_2)=1$, or  $c_1(v_1) \leq c_2(v_2)$ and for each successor $u_1$ of $v_1$ we have $M(u_1,v_2)=1$.
	\end{claim}
	
	Observe that once all values of $M$ are computed we can check if $P_1 \preceq P_2$ just by checking if for all $v_1$ there is $v_2$ such that $M(v_1,v_2)=1$.
	
	\begin{proof}[Proof of the claim]
		
		$M(v_1,v_2)=1$ iff there is a morphism from $F(\restr{P_1}{v_1})$ to $F(\restr{P_2}{v_2})$. In this morphism $F(v_1)$ is mapped either to $F(v_2)$, or to some other element of the tree $F(\restr{P_2}{v_2})$. The second case means that there is a successor $u_2$ of $v_2$, such that there is a morphism from $F(\restr{P_1}{v_1})$ to $F(\restr{P_2}{u_2})$. The first case means that the label of $c_1(v_1)$ is less or equal to $c_2(v_2)$, and that for any successor $u_1$ of $v_1$ there is a morphism of $F(\restr{P_1}{u_1})$ to $F(\restr{P_2}{v_2})$.
	\end{proof}
	
	We can bound the running time of the algorithm (up to a multiplicative constant) by the following expression:
	$$
	\sum_{v_1}\left( O(1) + d(v_1) + \sum_{v_2} \left( O(1) + d(v_2) + |A(c_1(v_1),c_2(v_2))| + d(v_1) + d(v_2)\right) \right),
	$$
	where $d(v)$ is the out-degree of $v$ and $|A(c_1(v_1),c_2(v_2))|$ is the running time of the algorithm $A(c_1(v_1),c_2(v_2))$ to compare the labels.
	In this expression the terms $O(1)+d(v)$ correspond to the time needed to explore vertex $v$ in depth-first search and the terms $d(v_1)$ and $d(v_2)$ in the end of the formula is the time needed to check the conditions of the claim (we need to go over all neighbors of $v_1$ and $v_2$ for this).
	
	Using the fact that the sum of the out-degrees of vertices in the graph is equal to the number of edges we can simplify the expression to the following (up to a multiplicative factor):
	\begin{align*}
	|V_1| + |E_1| + \sum_{v_1, v_2} \left( O(1) + d(v_2) + |A(c_1(v_1),c_2(v_2))| + d(v_1) + d(v_2)\right) = \\
	|V_1| + |E_1| + |V_1||V_2| + |V_1| |E_2|  + \sum_{v_1, v_2} A(c_1(v_1),c_2(v_2)) + |E_1| |V_2| + |V_1| |E_2| =\\	
	O\left(|V_1||V_2| + |V_1| |E_2| + |E_1| |V_2| + \sum_{v_1, v_2} A(c_1(v_1),c_2(v_2))\right)
	\end{align*}
	Denote by $a_0$ and $b_0$ the sizes of descriptions of $P_1$ and $P_2$ respectively without the descriptions of labels (the size of $P_i$ is $O(|V_i| + |E_i|)$). The sizes of descriptions of labels in $P_1$ we denote by $a_1, \ldots, a_k$, and in $P_2$ by $b_1, \ldots, b_l$. Note that by the statement of the theorem $A(c_1(v_1),c_2(v_2)) = O(a_i b_j))$, where $a_i$ is the size of $c_1(v_1)$ and $b_j$ is the size of $c_2(v_2)$. Then we can upper bound the running time of the algorithm (up to a multiplicative factor) by
	$$
	a_0 b_0 + \sum_{i=1}^k \sum_{j=1}^l a_i b_j \leq \left(\sum_{i=0}^{k} a_i \right) \cdot \left(\sum_{j=0}^{l} b_j \right)
	$$
	as needed.
\end{proof}

\section{Algorithms on Muller's $k$-acceptors}\label{accept}

Here we describe a feasible algorithm deciding the relation $L\leq_WM$ (meaning that $L=M\circ f$ for some continuous function $f$ on $X^\omega$), where the $\omega$-regular $k$-partitions $L,M:X^\omega\to\bar{k}$ are given by Muller $k$-acceptors recognizing them. 

The standard way to represent a Muller $k$-acceptor $({\mathcal M},A)$ is to describe the graph of $\mathcal{M}$ with $n$ vertices (corresponding to the states $q_1,\ldots,q_n$), $2^n$ bit vectors representing the sets of states, and labels $l<k$ of each vector representing the $k$-partition $A$. Because of the $2^n$ bit vector the size of this representation is $O(2^n)$ (we assume that $k$ is constant).
%, apparently it can not be considerably reduced in general. 

%\subsection{Main result on Muller's $k$-acceptors}
%\label{sec:muller}

%Under this representation we state the following.

\begin{theorem} \label{thm:main-acceptors}
The relation $L({\mathcal M}_1,A_1)\leq_WL({\mathcal M}_2,A_2)$ may be decided in quadratic time in the size of the inputs.
\end{theorem}

\begin{proof}
     According to the idea sketched in Sections \ref{in} and \ref{pre}, we first compute the levels $P_1,P_2\in\mathcal{V}_k(2)^\sqcup$ in which the given $k$-partitions are Wadge complete, and then apply the algorithm of Theorem \ref{pos} to check whether $P_1\preceq P_2$. We explain how to compute $P_1$. Let $\leq_0$ and $\leq_1$ be preorders on $C_{\mathcal{M}_1}$ defined in \cite{wag79} as follows: $c\leq_0d$, if some (equivalently, every) state in $d$ is reachable from some (equivalently, every) state in $c$; $c\leq_1d$, if $c\supseteq d$. It follows that $\leq_1$ implies $\equiv_0$ (the equivalence relation induced by $\leq_0$), and that any $\equiv_0$-equivalence class has a largest cycle under inclusion. The preorders are easily computable from the presentation of $({\mathcal M}_1,A_1)$ by using reachability in the graph of $\mathcal{M}_1$. Indeed, for preorder $\leq_0$ it is enough to precompute reachability relation in the graph of $\mathcal{M}_1$ (this can be done in time $O(n^2)$ for graphs of constant degree, for example, by running bredth-first search from each vertex of the graph) and then for each pair of subsets of vertices to check, if a vertex in one subset is reachable from a vertex in the other one. For this we need constant time for each pair of subsets, which gives us $O(2^{2n})$ overall.
     
     For relation $\leq_1$ we need to compute for each pair of subsets if one is included in the other one (we compute this relation for all subsets and then consider it for cycles only). We can think of subsets as enumerated by their characteristic vectors interpreted as binary representation of integers. We can compute the relation by splitting the vertices into two parts depending on the first coordinate of the characteristic vector. In each part we compute the relation recursively, and  to compute the relation between the parts, observe that the relation holds for $0x$ and $1y$ iff it holds for $1x$ and $1y$, and the latter is already computed. To bound the running time of this recursive procedure, note that for any pair of subsets we need constant time of computation. Thus, the total time of this step is $O(2^{2n})$ as well.
     
     %The time required by this computation is $O\left(n \cdot |({\mathcal M}_1,A_1)| \cdot |({\mathcal M}_2,A_2)|\right)$, since for each pair of cycles we need to compute reachability, which can be done in linear time in the size of the graph of $\mathcal{M}_1$, which is linear in $n$, since the graph has constant degree.\nb{The problem with this bound is that this is the most complicated computation now and our algorithm is not quadratic anymore}

    As explained in the proof of Lemma 15 in \cite{s23}, we can take $P_1=(C_{\mathcal{M}_1}/_{\equiv_0},\leq_0,d)$ where $d:C_{\mathcal{M}_1}/_{\equiv_0}\to\mathcal{V}_k$ is defined by $d([c]_0)=([c]_0,\leq_1,A_1|_{[c]_0})$. Note that the equivalence classes in $C_{\mathcal{M}_1}/_{\equiv_0}$ bijectively correspond to the reachable strongly connected components (SCCs) of the graph of $\mathcal{M}_1$ (i.e., the largest cycles), hence  $C_{{\mathcal M}_1}/_{\equiv_0}$ may be replaced by the set of all SCCs (canonical representatives in the equivalence classes).
    The latter set is computable in linear time from the graph of automaton by Tarjan's algorithm \cite{t72}. Altogether, $P_1,P_2$ are easily computable, and deciding the relation $P_1\preceq P_2$ is quadratic in the size of input by Theorem~\ref{pos}.
 \end{proof}

\section{Representation of $k$-acceptors} 
\label{sec:representation}

As we discussed, the standard representation of a Muller $k$-acceptor is of size $\Theta(2^n)$.
In this section we observe that actually, the number of possible cycles in the acceptor is at most $O(C^{n})$, where $C<2$ is a positive constant independent of the size of the acceptor (but dependent on the size of the alphabet). As a result, it is possible to have a more compact representation for $k$-acceptors.

Denote the number of vertices in the acceptor by $n$ and the size of the alphabet by $d$. We prove the following statement.

\begin{lemma} \label{lem:number-of-cycles}
    The number of cycles in an acceptor is at most
    $$
    \max(2^d, C^n+n),
    $$
    where $C = 2\left(1 - \frac{1}{2^{d+1}} \right)^{\frac{1}{d+1}} <2$. 
\end{lemma}

Note that the first term $2^d$ in the maximum is constant in terms of $n$ and is needed only to cover the case of small constant $n$.

In the proof of the lemma we will use the following lemma proved in~\cite{bhkk08} (we state only the special case of their lemma that will be enough for us). This lemma is a simple consequence of the classic Product Theorem~\cite[Theorem 22.10]{j11}.

\begin{lemma}[\cite{bhkk08}] \label{lem:product-theorem}
    Let $V$ be a finite set with $n$ elements and with subsets $A_1, \ldots, A_n$, such that every $v \in V$ is contained in exactly $\delta$ subsets. Let $\FF$ be a family of subsets of $V$ and assume that there is a log-concave function $f \geq 1$ such that the projections $\FF_i = \{F \cap A_i \mid F \in \FF \}$ satisfy $|\FF_i| \leq f(|A_i|)$ for each $i = 1, \ldots, n$. Then,
    $$
    |\FF| \leq f(\delta)^{n/\delta}.
    $$
\end{lemma}

Next we proceed to the proof of Lemma~\ref{lem:number-of-cycles}. The proof follows the same strategy as the proof of Lemma~6 in~\cite{bhkk08}.

\begin{proof}[Proof of Lemma~\ref{lem:number-of-cycles}]
    Denote the states of the automata by $Q = \{q_1, \ldots, q_n\}$ and consider the automata as a directed graph on $Q$, in which for every vertex $q_i$ and for every letter in the alphabet $x \in X$ there is an edge leaving $q_i$ labeled by $x$ (some edges might be parallel and some edges might be loops). In particular, the out-degree of every vertex in the graph is $|X|=d$.

    If $n \leq d$, then the number of cycles is less or equal to the number of subsets in $Q$, which is $2^d$. In this case we are done. Thus, from now on we  assume that $n \geq d+1$.

    Denote by $\inneigh(q)$ for $q \in Q$ the set of vertices that have an outgoing edge to $q$. Analogously denote by $\outneigh(q)$ the set of vertices that have incoming edges from $q$. Let $V = Q$ and $A_i = \{q_i\} \cup \{\inneigh(q_i)\}$. If for $q_j$ we have $|\outneigh(q_j)| = l$, then $q_j$ is contained in at most $l+1$ sets $A_i$. Note that $l$ might be smaller than $d$, since there might be parallel edges and loops. If $l < d$, we add $q_j$ to arbitrary sets $A_i$, that do not contain it yet, until $q_j$ is in $d+1$ sets. This is possible, since the total number of sets is $n$ and $n \geq d+1$. Denote the resulting sets by $A_i'$. Now each vertex $q_j$ is contained in exactly $d+1$ sets $A_i'$. 

    Denote by $\CC$ the set of all cycles in the acceptor. Let $\CC' = \CC \setminus \{\{q\} \mid q \in Q\}$, that is $\CC'$ consists of cycles of size at least 2.

    Note, that for any $i$ the size of the set of projections $|\CC_i'| = \{C \cap A_i \mid C \in \CC' \}$ is at most $2^{|A_i|}-1$, since $\{q_i\}$ cannot be a projection (any cycle of size at least $2$ containing $q_i$ must contain one of its neighbors in $\inneigh(q_i)$).

    Consider a log-concave function $f(a) = 2^{a} -1 $ and apply Lemma~\ref{lem:product-theorem} to $f$, $\CC'$, sets $A_1', \ldots, A_n'$ and $\delta = d+1$. We get
    $$
    |\CC'| \leq f(d+1)^{n/d+1} = \left( 2^{d+1} - 1 \right)^{\frac{n}{d+1}}.
    $$
    Since $|\CC| \leq |\CC'| + n$, the lemma follows.
\end{proof}

Thus, the number of cycles can be substantially smaller than the number of subsets on $Q$. As a result, if in the representation of an acceptor we provide the binary vector that contains a bit for each cycle, instead of each subset of states, the representation becomes more compact. However, for an algorithm to interpret this input, it needs first to compute the list of all cycles. In the next lemma we show that this can be done efficiently.

\begin{lemma}
    Given an acceptor $\mathcal M$ the set $\CC$ of all cycles can be computed in time $O(n \cdot |\CC|^2)$.
\end{lemma}

\begin{proof}
    Let an \emph{elementary cycle} in a directed graph be a sequence of distinct vertices $v_1, \ldots, v_t$, such that from each vertex $v_i$ there is a directed edge to $v_{i+1}$ and there is a directed edge from $v_t$ to $v_1$.
    The paper~\cite{j75} provides an algorithm that constructs all elementary cycles in a directed graph $G = (V,E)$ in time $O((|V|+|E|)p)$, where $p$ is the number of elementary cycles. We first apply this algorithm to the acceptor and denote the resulting list by $\CC'$.

    From this list we can construct the list $\CC$. For this we first consider an undirected graph $G_C$ with $\CC'$ as a set of vertices, such that $C_i$ and $C_j$ are connected by an undirected edge if they share at least one vertex. It is easy to see that cycles in $\CC$ correspond to induced connected subgraphs of $G_C$.

    The graph $G_C$ can be constructed in time $O(n |\CC'|^2)$: it is enough for each pair of vertices $C_i$, $C_j$ to check if the cycles share a vertex. This can be done in $O(n)$.

    To store all elements of $\CC$ we maintain a red-black tree data structure on them. In the beginning the data structure is empty. To add elements of $\CC$ to the red-black tree we iterate through the induced connected subgraphs of $G_C$ and check if they correspond to a new element of $\CC$. More precisely, first we create a queue of size 1 subgraphs of $G_C$ (they correspond to just elements of $\CC'$). With each element $S \subseteq \CC'$ in the queue we will store the characteristic vector $N(S)$ of the set of its neighbors in $G_C$ and the characteristic vector $Cyc(S)$ of the cycle in $\mathcal M$ that it corresponds to. Computing these vectors for the each of the first elements of the queue naively takes time $O(|\CC'|)$ and $O(n)$ respectively. 

    Next we repeat the following step. We extract the current first element $S \subseteq \CC'$ of the queue. For each of its neighbors $C$ (there are at most $|\CC'|$ of them), we add $C$ to $S$ and compute $Cyc(S \cup C)$. This can be done in time $O(n)$ by scanning through the corresponding vectors for $S$ and $C$. We check if the cycle $Cyc(S \cup C)$ was already computed before. If it was, we move on to the next neighbor of $S$. If this is a new cycle we add it to the red-black tree and add $S \cup C$ to the end of the queue. These operations with the red-black tree can be performed in time logarithmic in the size of the tree, that is in time $O(\log |\CC|) = O(n)$ (since $|\CC| \leq 2^n$). When adding new element to the queue we compute $N(S \cup C)$, this takes time $O(|\CC'|)$.  We are done once the queue is empty.

    For the correctness of this procedure, note that if two induced connected subgraphs $S, S' \subseteq \CC'$ correspond to the same cycle, they have the same set of neighbors in $G_C$. As a result, our queue will scan through subgraphs corresponding to all cycles. 
    
    Moreover, for each cycle $C \in \CC$ we will have exactly one subgraph corresponding to it. For the running time this means that the total number of subgraphs that are added to the queue is $|\CC|$. When we add each element to the queue, we compute $N(S)$ for it. For each $S$ in the queue we consider at most $|\CC'|$ neighbors, compute $Cyc$ for the union and perform queue operations. In total, the running time is $O(|\CC| \cdot (|\CC'| + |\CC'|\cdot n)) = O(|\CC| \cdot |\CC'| \cdot n) = O(|\CC|^2 \cdot n)$.
\end{proof}

Next we prove a version of Theorem~\ref{thm:main-acceptors} for the more compact representation of $k$-acceptors.

\begin{theorem}
The relation $L({\mathcal M}_1,A_1)\leq_WL({\mathcal M}_2,A_2)$ may be decided in time $O\left(n \cdot C^2 + n^2\right)$, where $n$ is the size of the representations of graphs of ${\mathcal M}_1$ and ${\mathcal M}_2$ and $C$ is the size of the representations of cycles in ${\mathcal M}_1$ and ${\mathcal M}_2$.
\end{theorem}

\begin{proof}
    The proof is analogous to the proof of Theorem~\ref{thm:main-acceptors}.
    
    As before, we compute preorders $\leq_0$ and $\leq_1$. 
    For $\leq_0$, as before, we can precompute reachability relation on the graph of the automata (this takes $O(n^2)$ time) and then for each pair of cycles check reachability between a couple of vertices in them (this requires $C^2$ time). For the inclusion relation we can check inclusion of each pair of cycles in $O(n)$ straightforwardly. This results in time $O(n \cdot C^2)$

    The remaining part of the proof remains completely the same and requires $O(C^2)$ time.
 \end{proof}
%\section{Conclusion}\label{con}

\subsection*{Acknowledgments}
%\subsubsection{\ackname} 
V. Selivanov's research was supported by the Russian Science Foundation, project 23-11-00133.

\nocite{*}
\bibliographystyle{eptcs}
\bibliography{bib}

\end{document}